\documentclass[final,5p,times,twocolumn]{elsarticle}

\usepackage[T1]{fontenc}
\usepackage{soul}

\usepackage{siunitx}
\usepackage{slashed}
\usepackage{enumitem}
\usepackage{xcolor, color}

\usepackage{multirow}
\usepackage{nccmath}
\usepackage{cite}
\usepackage{amsfonts}
\usepackage{amsmath}
\usepackage{amssymb}
\usepackage{mathtools}
\usepackage{mathrsfs}
\usepackage{graphicx}
\usepackage{subcaption}
\usepackage{algpseudocode}
\usepackage{booktabs}
\usepackage{adjustbox}
\usepackage{xparse}
\usepackage[hidelinks]{hyperref}
\usepackage{bm}

\NewDocumentCommand{\basis}{om}{%
	\IfNoValueTF{#1}
	{\| \Phi^{#2}\|^2}
	{\langle \Phi^{#1},\phi^{#2} \rangle}%
}

\newcommand{\da}{\mathrm{d}}
\newcommand{\ud}{\neg\mathrm{d}}

\newcommand{\rank}{\mathrm{rank}}
\newcommand{\ik}[1]{\mathcal{I}_{#1}^v}
\newcommand{\pce}[1]{\mathsf{#1}}
\newcommand{\tra}[3]{#1_{[#2,#3]}}
\newcommand{\trad}[3]{#1_{[#2,#3]}^{\mathrm{d}}}
\newcommand{\traud}[3]{#1_{[#2,#3]}^{\neg\mathrm{d}}}

\newcommand{\ini}{\text{ini}}
\newcommand{\tini}{T_\text{ini}}

\newcommand{\mbb}[1]{\mathbb{#1}}
 
\newcommand{\mcl}[1]{\mathcal{#1}}

\newcommand{\splk}[2]{\mcl{L}^2(\Omega, \mathcal{F}_{#1}, \mu; \mathbb{R}^{#2})}
\newcommand{\splx}[1]{\mcl{L}^2(\Omega, \mathcal{F}, \mu; \mathbb{R}^{#1})}
\newcommand{\lsp}{\mcl{L}^2} 

\newcommand{\diff}{\mathop{}\!\mathrm{d}}

\newcommand{\mean}{\mbb{E}}
\newcommand{\var}{\mbb{V}}

\newcommand{\dimy}{{n_y}}
\newcommand{\dimx}{{n_x}}
\newcommand{\dimu}{{n_u}}
\newcommand{\dimw}{{n_w}}
\newcommand{\dimz}{{n_z}}

\newcommand{\I}{\mathbb{I}}
\newcommand{\N}{\mathbb{N}}
\newcommand{\R}{\mathbb{R}}

\newcommand{\Hankel}{\mcl{H}}

\newcommand{\End}{\hfill $\square$}

\usepackage{amsthm}
\newtheorem{cor}{Corollary}
\newtheorem{lem}{Lemma}

\newtheorem{defn}{Definition}
\newtheorem{rem}{Remark}
\newtheorem{assum}{Assumption}

\newtheorem{exmp}{Example}

\journal{Journal of \LaTeX\ Templates}
\biboptions{authoryear}

\begin{document}
\begin{frontmatter}
\title{A Stochastic Fundamental Lemma with Reduced Disturbance Data Requirements}

\author{Ruchuan Ou}\ead{ruchuan.ou@tuhh.de}
\author{Guanru Pan}\ead{guanru.pan@tuhh.de}  

\author{Timm Faulwasser\corref{mycorrespondingauthor}}
\cortext[mycorrespondingauthor]{Corresponding author}
\ead{timm.faulwasser@ieee.org}

\address{Institute of Control Systems, Hamburg University of Technology, 21079 Hamburg, Germany}                                        

\begin{keyword}                            
Fundamental lemma, stochastic systems, polynomial chaos, reduced disturbance data, non-Gaussian uncertainties, data-driven control
\end{keyword}

\begin{abstract}
Recently, the fundamental lemma by Willems et al. has been extended towards stochastic LTI systems subject to process disturbances. Using this lemma requires previously recorded data of inputs, outputs, and disturbances. In this paper, we exploit causality concepts of stochastic control to propose a variant of the stochastic fundamental lemma that does not require past disturbance data in the Hankel matrices. Our developments rely on polynomial chaos expansions and on the knowledge of the disturbance distribution. Similar to our previous results, the proposed variant of the fundamental lemma allows to predict future input-output trajectories of stochastic LTI systems. We draw upon a numerical example to illustrate the proposed variant in data-driven control context.
\end{abstract}
\end{frontmatter}

\section{Introduction} \label{sec:Introduction}
Data-driven system representations based on the fundamental lemma by \citet{willems05note} are of continued and increasing research interest. The main insight is that the trajectories of any controllable Linear Time Invariant (LTI) systems can be characterized without explicit identification of a state-space model. Specifically, under the assumption on the persistency of excitation, any finite trajectory of an LTI system lives in the column space of a Hankel matrix constructed from the recorded input-output data. 
Beyond the deterministic LTI setting, there are some variants of the lemma, e.g., extensions to nonlinear systems \citep{alsalti21data,lian21koopman}, to linear parameter-varying systems \citep{verhoek21data}, to kernel representation of nonlinear systems \citep{molodchyk24exploring}. Especially, some recent works have extended the fundamental lemma towards stochastic predictive control schemes, see \citet{kerz23data} for a tube-based approach, \citet{teutsch24sampling} for a sampling-based approach, \citet{chiuso25harnessing} for a moment-based approach via a separation principle, and \citet{breschi23data} for a unified framework using regularization. We refer to \citet{markovsky21behavioral} for an overview and refer to \citet{willems13open} for early discussions of behavioral concepts for stochastic systems.

In previous work, we proposed a stochastic variant of the fundamental lemma for stochastic LTI systems subject to process disturbances in the framework of Polynomial Chaos Expansions (PCE) \citep{pan23stochastic}.
PCE dates back to \citet{wiener38homogeneous} and has been generalized to a wide class of non-Gaussian distributions by \citet{xiu02wiener}. The core idea of PCE is that one can use suitable orthogonal bases of an $\lsp$  Hilbert space to model any square integrable random variable via a series expansion. Therefore, $\lsp$ random variables can be parameterized linearly by deterministic coefficients in appropriately chosen polynomial bases. Indeed, one may represent any  $\lsp$ random variable with an affine PCE with two terms. We refer to \citet{sullivan15introduction} for a general introduction and to \citet{kim13wiener} for an early overview on control design using PCE. PCE was first introduced in stochastic optimal control to foster computation \citep{fagiano12nonlinear,paulson14fast}, it also has prospects for system theory and analysis \citep{paulson15stability, ahbe20region, faulwasser23behavioral}.

The stochastic fundamental lemma by \citet{pan23stochastic} requires measurements of the past disturbances to construct the Hankel matrices. As such measurements are, in general, costly or difficult to obtain, the main goal of this paper is to obtain a modified stochastic fundamental lemma that does not require past disturbance data in the online computation. Using causality properties, we show that a stochastic LTI system can be decomposed into a number of decoupled stochastic subsystems, each of which captures the effect of one source of stochastic uncertainty and how it propagates via the dynamics. In each subsystem, we convert the additive process disturbance into an output initial condition using causality such that the subsystem is not subject to disturbances. Combining this with PCE, we derive corresponding data-driven representations of the PCE reformulated subsystems without disturbance data in the Hankel matrices under mild conditions. Moreover, we provide a modified stochastic fundamental lemma in random variables based on the results in PCE coefficients. Drawing upon an example, we demonstrate the numerical advantages of the modified approach compared to the stochastic fundamental lemma for data-driven stochastic Optimal Control Problems (OCP). Importantly, this work should not be understood as a means of overcoming or replacing subspace identification or the stochastic fundamental lemma by \citet{pan23stochastic}, but rather as a complementary approach to existing tools that is closely linked to stochastic extensions of behavioral systems theory. The main focus of this paper is on uncertainty propagation, i.e., the computation of random variable trajectories of stochastic LTI systems, and the use of our results for control is beyond the scope.

The remainder is structured as follows: In Section~\ref{sec:Preliminaries} we recall the setting and problem statement, while in Section~\ref{sec:MainResults} we present the modified stochastic fundamental lemma that does not require past disturbance data in the online computation. Section~\ref{sec:Simulation} illustrates a numerical example. The paper ends with conclusions in Section~\ref{sec:Conclusion}.

\subsubsection*{Notation}
Let $(\Omega,\mathcal F,\mu)$ be a probability space, where $\Omega$ is the set of possible outcomes, $\mathcal{F}$ is a $\sigma$-algebra, and $\mu$ is the considered probability measure.
$\splx{n_z}$  is the space of vector-valued random variables of finite covariance and of dimension $n_z$.
Let $Z: \I_{[0,T-1]} \rightarrow \splx{n_z}$ be a sequence of vector-valued random variables from time instant 0 to $T-1$. We denote by $\mean[Z]$ and $z\coloneqq Z(\omega) \in \R^{n_z}$ its mean and realization, respectively. The vectorizations of $z$ and $Z$ are written as $z_{[0,T-1]} \coloneqq [z_0^\top,z_1^\top, \dots,z_{T-1}^\top]^\top \in \R^{n_z T}$ and $Z_{[0,T-1]}\in\splx{n_zT}$, respectively.
We denote the identity matrix of size $n$ by $I_n$.
For any matrix $Q\in \mbb{R}^{n\times m}$ with columns $q^1, \dots, q^m$, the column-space is denoted by $\mathrm{colsp}(Q) \coloneqq \mathrm{span}\left(\{q^1,\dots, q^m\}\right)$.

\section{Preliminaries} \label{sec:Preliminaries}
We consider stochastic discrete-time LTI systems
\begin{subequations} \label{eq:Dyn}
	\begin{align}
		X_{k+1} &= AX_k + BU_k + EW_k,\quad X_0=\tilde{X}_0,\\
		Y_k &= CX_k\label{eq:Output}
	\end{align}
\end{subequations}
with state $X_k \in\splk{k}{n_x}$, $X_k:\Omega\to\R^{n_x}$ and process disturbance $W_k \in \splx{n_w}$.
Measurement noise is not included in~\eqref{eq:Dyn}, i.e., the measurements of the output are exact.

In the filtered probability space $(\Omega, \mcl F, (\mcl F_k)_{k\in \N}, \mu)$, $(\mcl F_k)_{k\in \N}$ is the smallest filtration that the stochastic process $X$ is adapted to. That is, $\mcl F_k = \sigma(X_i,i\leq k)$, where $\sigma(X_i,i\leq k)$ denotes the $\sigma$-algebra generated by $X_i,i\leq k$. Moreover, $\mcl F$ is the $\sigma$-algebra that contains all available historical information, i.e., $\mcl F_0 \subseteq \mcl F_1 \subseteq ...  \subseteq \mcl F$. We model $Y_k$ and $U_k$, $k\in\N$ as stochastic processes adapted to the filtration $(\mathcal{F}_k)_k$, i.e. $Y_k \in\splk{k}{\dimy}$ and $U_k \in\splk{k}{\dimu}$. This immediately imposes a causality constraint on $Z_k$, $Z\in\{U,X,Y\}$, i.e., $Z_k$ only depends on past disturbances $W_j$, $j< k$. It directly follows from Lemma 1.14 by \citet{kallenberg21foundations} that inputs $U_k$, $k\in\N$ adapted to the filtration $\mcl F_k$ are equivalent to feedback polices based on historical state information. For more details on filtrations we refer to \citet{fristedt13modern}. For the sake of readability, the space $\splk{k}{n_z}$ is written as $\lsp_{k}(\R^\dimz)$.

For a specific outcome $\omega \in \Omega$, the realization $W_k(\omega)$ is denoted as $w_k$. Likewise, we write $u_k \coloneqq U_k(\omega)$, $x_k \coloneqq X_k(\omega)$, and $y_k \coloneqq Y_k(\omega)$. Given a specific initial condition $\tilde{x}_0$ and a sequence of disturbance realizations $w_k$, the realization dynamics induced by \eqref{eq:Dyn} are
\begin{subequations} \label{eq:DynReal}
	\begin{align} 
		x_{k+1} &= Ax_k + Bu_k +  Ew_k,\quad x_0=\tilde{x}_0,\\
		y_k &= Cx_k.
	\end{align}
\end{subequations}
Then we make the following assumption that was also adopted in the original work of \citet{willems05note}.
\begin{assum}[System properties] \label{ass:Sys}
	\quad We assume $(A,B)$ controllable and $(A,C)$ observable in system~\eqref{eq:DynReal}.
\end{assum}
If system~\eqref{eq:DynReal} is observable, the lag $\ell$ is the minimum integer such that the observability matrix $\mcl{O}_{\ell}\coloneqq [C^\top, (CA)^\top,...,(CA^{\ell-1})^\top]^\top$ is of full column rank, i.e., $\text{rank}(\mcl{O}_{\ell})$ $=n_x$. Thus, for any $\tini\in\N$ satisfying $\tini\geq\ell$, $\text{rank}(\mcl{O}_{\tini})$ $=n_x$ holds.

\subsection{Polynomial Chaos Expansion}
The key idea of PCE is that any $\lsp$ random variable can be expressed in a suitable (orthogonal) polynomial basis. Consider an orthogonal polynomial basis $\{\phi^j(\xi)\}_{j=0}^{\infty}$ that spans the space $\mcl{L}^2(\Omega, \mathcal{F}, \mu; \mathbb{R})$, i.e.,
\begin{equation} \label{eq:Orthogonality}
	\langle \phi^i(\xi),\phi^j(\xi) \rangle = \int_{\Omega} \phi^i(\xi) \phi^j(\xi) \diff \mu = \delta^{ij}\langle\phi^j(\xi) \rangle^2,
\end{equation}
where $\xi\in \lsp(\Omega, \mathcal{F}, \mu; \Xi)$ is the argument of polynomials, $\langle \phi^j(\xi) \rangle^2\coloneqq \langle \phi^j(\xi),\phi^j(\xi) \rangle$, and $\delta^{ij}$ is the Kronecker delta. Note that $\xi:\Omega\to\Xi$ is a function of the outcome $\omega$ and hence $\phi^j(\xi(\omega)) = \phi^j\circ\xi(\omega)$. The first polynomial is always chosen to be $\phi^0(\xi) = 1$. Hence, the orthogonality~\eqref{eq:Orthogonality} gives that for all other basis dimensions $j>0$, we have $\mean[\phi^j(\xi)]=\int_{\Omega} \phi^j(\xi) \diff \mu= \langle \phi^j(\xi),\phi^0(\xi)\rangle=0$.

The PCE of a real-valued random variable $Z \in  \lsp(\R)$ with respect to the basis $\{\phi^j(\xi)\}_{j=0}^{\infty}$ is 
\[
Z(\omega) = \sum_{j=0}^{\infty}\pce{z}^j \phi^j(\xi(\omega)) \quad \text{with} \quad \pce{z}^j = \frac{\big\langle Z(\omega), \phi^j(\xi(\omega)) \big\rangle}{\langle\phi^j(\xi) \rangle^2},
\]
where $\pce{z}^j\in \R$ is referred to as the $j$-th PCE coefficient.
For the sake of readability, we omit the arguments $\xi(\omega)$, $\omega$ whenever there is no ambiguity. The PCE of a vector-valued random variable $Z\in\lsp(\R^\dimz)$ follows by applying PCE component-wise, i.e., the $j$-th PCE coefficient of $Z$ reads
$\pce{z}^{j} =\begin{bmatrix} \pce{z}^{1,j} & \pce{z}^{2,j} & \cdots & \pce{z}^{n_z,j} \end{bmatrix}^\top$, where $\pce{z}^{i,j}$ is the $j$-th PCE coefficient of $i$-th component $Z^i$, $\forall i\in\I_{[1,n_z]}$. Moreover, as the first basis function is always chosen to be 1, i.e. $\phi^0=1$, we have $\pce{z}^0=\mean[Z]$ from the definition of PCE.

In numerical implementations the expansions have to be truncated after a finite number of terms. This may lead to truncation errors
\[
	\Delta Z(L) = Z -  \sum\nolimits_{j=0}^{L-1}\pce{z}^j \phi^j,
\]
where $L\in\N^\infty\coloneqq \N^+\cup\{\infty\}$ is the PCE dimension. For $L\to\infty$, the truncation error satisfies $\lim_{L\to\infty}\|\Delta Z(L)||=0$ \citep{cameron47orthogonal, ernst12convergence}. 

\begin{defn}[Exact PCE representation] \label{def:ExactPCE}
	The PCE of a random variable $Z \in \lsp(\R^\dimz)$ is said to be exact with dimension $L\in\N^\infty$ if $ Z -  \sum_{j=0}^{L-1}\pce{z}^j \phi^j=0$.
\end{defn}

\begin{rem}[Generic affine PCE series ] 
	Given an $\mcl{L}^2$ random variable with known distribution, the key to construct an exact finite-dimensional PCE is the appropriate choice of basis functions. For some widely used distributions, the appropriate choice of polynomial bases is summarized in \citet{xiu02wiener}. Additionally, a generic (non-orthonormal but orthogonal) basis choice for any random variable $Z \in \lsp(\R)$ is $\phi^0 = 1$ and $\phi^1 = Z-\mean[Z]$, which implies the exact and finite PCE $\pce{z}^0 = \mean[Z]$ and $\pce{z}^1 = 1$.
\end{rem}

Replacing all random variables in~\eqref{eq:Dyn} with corresponding PCE representations with respect to the basis $\{\phi^j\}_{j=0}^{\infty}$, we get
	\begin{equation*}
		\sum\nolimits_{j=0}^\infty \pce{x}_{k+1}^j \phi^j= \sum\nolimits_{j=0}^\infty (A\pce{x}_k^j + B\pce{u}_k^j + E\pce{w}_k^j)\phi^j.
	\end{equation*}
Then for all $j\in\N^\infty$, performing Galerkin projection onto the basis functions $\phi^j$, one obtains the dynamics of the PCE coefficients
\begin{subequations} \label{eq:DynPCE}
	\begin{align}	
		\pce{x}_{k+1}^j &= A\pce{x}_k^j + B\pce{u}_k^j + E\pce{w}_k^j, \quad \pce{x}_0^j=\tilde{\pce{x}}_0^j,\\
		\pce{y}_k^j &= C\pce{x}_k^j.
	\end{align}
\end{subequations}
Note that dynamics~\eqref{eq:Dyn}, \eqref{eq:DynReal} and \eqref{eq:DynPCE} have the identical system matrices as the Galerkin projection preserves the linearity of the system. For details about Galerkin projection we refer to Appendix~A of \citet{pan23stochastic}.

\subsection{Stochastic Fundamental Lemma}
Data-driven system representation based on \citet{willems05note} requires persistently exciting input data $\tra{u}{0}{T-1}\coloneqq [u_0^\top,u_1^\top,...,u_{T-1}^\top]^\top$.
\begin{defn}[Persistency of excitation] Let $T$, $N \in \N^+$. A sequence of real-valued data $\tra{u}{0}{T-1}$ is said to be persistently exciting of order $N$ if the Hankel matrix
	\begin{equation*}
		\Hankel_N(\tra{u}{0}{T-1}) \doteq \begin{bmatrix}
			u_0   &\cdots& u_{T-N} \\
			\vdots & \ddots & \vdots \\
			u_{N-1}& \cdots  & u_{T-1} \\
		\end{bmatrix}
	\end{equation*}
	is of full row rank.
\end{defn}

As we are interested in the stochastic LTI system~\eqref{eq:Dyn}, one could attempt using a persistently exciting input sequence $\tra{U}{0}{T-1}$ of $\mcl{L}^2$ random variables to represent the stochastic LTI system \eqref{eq:Dyn}.
We have shown in \citet{pan23stochastic} that input-output-disturbance realization trajectories $\trad{(u,y,w)}{0}{T-1}$, where the superscript $\cdot^{\da}$ denotes offline recorded data of the disturbed system~\eqref{eq:DynReal}, suffices to represent \eqref{eq:Dyn}.

\begin{lem}[Stochastic fundamental lemma]\label{lem:StochFundam}~
    Consider system~\eqref{eq:Dyn} and a $T$-length realization trajectory tuple $\trad{(u,y,w)}{0}{T-1}$ of realization dynamics~\eqref{eq:DynReal}. We assume that system~\eqref{eq:DynReal} is controllable and let $\trad{(u,w)}{0}{T-1}$ be persistently exciting of order $N+\dimx$. Then $ \tra{(U,Y,W)}{0}{N-1}$ is a trajectory of \eqref{eq:Dyn} of length $N$ if and only if there exists $G \in \splx{T-N+1} $ such that 
	\begin{equation*}
		\Hankel_N(\trad{z}{0}{T-1}) G=Z_{[0,N-1]}
	\end{equation*} 
	holds for all $(z, Z)\in \{ (u, U), (y, Y), (w, W)\}$.
\end{lem}

\begin{cor}\label{coro:StochFundamPCE}
	Let the conditions of Lemma~\ref{lem:StochFundam} hold. Then $ \tra{(\pce{y}^j, \pce{u}^j, \pce{w}^j)}{0}{N-1}$, $j \in \N^\infty$ is a trajectory of the dynamics of PCE coefficients \eqref{eq:DynPCE} if and only if there exists $\pce{g}^j \in \R^{T-N+1}$ such that 
	\begin{equation*}
		\Hankel_t(\trad{z}{0}{T-1}) \pce{g}^j =\pce{z}^j_{[0,N-1]},~j\in \N^\infty
	\end{equation*} 
	holds for all $(z, \pce{z})\in \{(u,\pce{u}), (y,\pce{y}), (w, \pce{w})\}$.
\end{cor}
In the image representation of a fundamental lemma, the future trajectory of an LTI system lives in the column space the Hankel matrices consisting of the past trajectory data. In \citet{pan23stochastic}, however, we have shown that Hankel matrices in random variables cause conceptual issues and the corresponding fundamental lemma does not necessarily hold. Instead, one can predict the future trajectory of the stochastic LTI system~\eqref{eq:Dyn} from the data of the system in realizations~\eqref{eq:DynReal}.
\section{Forward uncertainty propagation with reduced data requirements} \label{sec:MainResults}
As the realizations of process disturbances are in general difficult to obtain, we aim for a variant of Lemma~\ref{lem:StochFundam} that does not explicitly rely on past disturbance data for the online computation. To this end, we will exploit the superposition principle and causality requirements of stochastic systems.

\subsection{VARX model}
To avoid the use of state measurements, we switch to the Vector AutoRegressive with eXogenous input (VARX) model structure
\begin{multline} \label{eq:VARXEW}
	Y_{k} = \hat{A}\tra{Y}{k-\tini}{k-1} + \hat{B}\tra{U}{k-\tini}{k-1}+\\
	 \hat{E}\tra{W}{k-\tini}{k-1}.
\end{multline}
\citet{sadamoto23equivalence} shows the equivalence between the above VARX model and the state-space model~\eqref{eq:Dyn} with an observability assumption. Moreover, \citet{sadamoto23equivalence} explicitly derives the relation between matrices $A$, $B$, $E$ of \eqref{eq:Dyn} and $\hat{A}$, $\hat{B}$, $\hat{E}$ of \eqref{eq:VARXEW} as
\begin{align*}
	\hat{A} &= CA^{\tini}\mcl{O}^{\dagger},\\
	\hat{B} &= C\left(\mcl{C}(B)-A^{\tini}\mcl{O}^{\dagger}\mcl{M}(B) \right),\\
	\hat{E} &= C\left(\mcl{C}(E)-A^{\tini}\mcl{O}^{\dagger}\mcl{M}(E) \right),
\end{align*}
where $\cdot^\dagger$ denotes the Moore-Penrose inverse. The controllability matrix $\mcl{C}$ and the matrix $\mcl{M}$ are defined as $\mcl{C}(D) \coloneqq \begin{bmatrix} A^{\tini-1}D & \cdots & AD & D\end{bmatrix}$ and $\mcl{M}(D) \coloneqq \begin{bmatrix} 0 \\ CD & \ddots \\ \vdots & \ddots & &\ddots \\ CA^{\tini-2}D & \cdots & &CD & 0 \end{bmatrix}$ for $D\in\{B,E\}$.
Note that the mapping between matrices $A$, $B$, $E$ and $\hat{A}$, $\hat{B}$, $\hat{E}$ is not unique, e.g., \citet{phan96relationship} provide another possibility. The equivalence between the state-space model and the VARX model suggests that Lemma~\ref{lem:StochFundam} and Corollary~\ref{coro:StochFundamPCE} also hold for the VARX model \citep{pan25data}.

It is common in behavioral systems theory to avoid assumptions on the structure of state-space and VARX models as much as possible. Hence, we model the effect of the process disturbance on the dynamics at each time step via a single aggregated term: $W_k$ in state-space models $X_{k+1}=AX_k+BU_k+EW_k$, where $E$ is often assumed to be an identity matrix, and $V_{k-1}$ in VARX models.
That is, the VARX model is reformulated as
\begin{subequations} \label{eq:VARX}
	\begin{align}
		&Y_k = \hat{A}\tra{Y}{k-\tini}{k-1} + \hat{B}\tra{U}{k-\tini}{k-1} + V_{k-1},\\
		&\tra{(U,Y)}{1-\tini}{0} = \tra{(\tilde{U},\tilde{Y})}{1-\tini}{0},
	\end{align}
\end{subequations}
where $V_{k-1}\coloneqq\hat{E}\tra{W}{k-\tini}{k-1}$, $\tra{(\tilde{U},\tilde{Y})}{1-\tini}{0}$ is an initial input-output trajectory of length $\tini$. For simplicity and ease of notation, we assume that the disturbances $V_k$, $k\in\N^\infty$ are i.i.d. random variables.
This is similar to the widely accepted assumption in state-space models that disturbances are i.i.d. We will discuss its generalization as well as stochastic LTI systems with measurement noise in Section~\ref{sec:Extension}. 
\begin{assum}[I.i.d. disturbances] \label{ass:iid}
	The disturbances~$V_k$, $k\in\N^\infty$ in VARX model~\eqref{eq:VARX} are i.i.d. random variables. Moreover, the distribution is known.
\end{assum}

Similar to~\eqref{eq:DynPCE}, we obtain the corresponding dynamics of realizations
\begin{subequations} \label{eq:VARXReal}
	\begin{align}
		&y_k = \hat{A}\tra{y}{k-\tini}{k-1} + \hat{B}\tra{u}{k-\tini}{k-1} + v_{k-1},\\
		&\tra{(u,y)}{1-\tini}{0} = \tra{(\tilde{u},\tilde{y})}{1-\tini}{0},
	\end{align}
\end{subequations}
and the dynamics of PCE coefficients for~\eqref{eq:VARX} for all $j\in\I_{[0,L-1]}$, $k\in\I_{[1,N]}$
\begin{subequations} \label{eq:VARXPCE}
	\begin{align}
		&\pce{y}_k^j = \hat{A}\tra{\pce{y}^j}{k-\tini}{k-1} + \hat{B}\tra{\pce{u}^j}{k-\tini}{k-1} + \pce{v}^j_{k-1},\\
		&\tra{(\pce{u},\pce{y})^j}{1-\tini}{0} = \tra{(\tilde{\pce{u}},\tilde{\pce{y}})^j}{1-\tini}{0}.
	\end{align}
\end{subequations}
Especially, given a deterministic initial trajectory $\tra{(\tilde{u},\tilde{y})}{1-\tini}{0}$, we get $\mean[\tra{(\tilde{u},\tilde{y})}{1-\tini}{0}]=\tra{(\tilde{u},\tilde{y})}{1-\tini}{0}$ and thus its PCE representation reads
\begin{align*}
	\tra{(\tilde{\pce{u}},\tilde{\pce{y}})^j}{1-\tini}{0} = \begin{cases}
		\tra{(\tilde{u},\tilde{y})}{1-\tini}{0},\quad &\text{for } j=0\\
		0,&\text{otherwise}
	\end{cases}.
\end{align*}
We also note that the dynamics of realizations and of PCE coefficients, i.e.~\eqref{eq:VARXReal} and~\eqref{eq:VARXPCE}, share the same system matrices. Therefore, the Hankel matrices consisting of realizations and PCE coefficients span the same column space provided persistency of excitation holds. This way, one can use the Hankel matrices in realizations to predict the future trajectories of PCE coefficients, see the column-space equivalence given by Lemma~3 by~\citet{pan23stochastic} for details.

\subsection{System decomposition}
For the sake of simplified notation, we consider deterministic initial input-output condition $\tra{(U,Y)}{1-\tini}{0}=\tra{(\tilde{u},\tilde{y})}{1-\tini}{0}$. We discuss the generalization to an uncertain initial condition in Remark~\ref{rem:UncertainIni}.
	
According to the superposition principle, we decouple the stochastic system~\eqref{eq:VARX} into subsystems based on the different sources of uncertainties. As an independent disturbance $V_k$ enters the system~\eqref{eq:VARX} at each time step $k\in\N$ and affects the future inputs and outputs, we obtain the decomposition
\begin{subequations} \label{eq:Superposition}
\begin{equation}
	Z_k = \mean[Z_k] + \sum_{i=0}^{k-1} Z_k^{v_i},\quad Z\in\{U,Y\},
\end{equation}
where $\mean[Z_k]$ is the nominal part and $Z_k^{v_i}$ denotes the effect of disturbance $V_i$ on $Z_k$. Therefore, the dynamics of the nominal deterministic subsystem are
\begin{equation} \label{eq:SubNom}
	\begin{split}
		&\mean[Y_k] = \hat{A}\mean[\tra{Y}{k-\tini}{k-1}] + \hat{B}\mean[\tra{U}{k-\tini}{k-1}] +\mean[V],\\
		&\mean[\tra{(U,Y)}{1-\tini}{0}] = \tra{(\tilde{u},\tilde{y})}{1-\tini}{0},
	\end{split}
\end{equation}
where $\mean[V]$ is the expected value for i.i.d. disturbances. The stochastic error system for each $V_i$, $i\in\N$ reads
\begin{equation} \label{eq:SubError}
	\begin{split}
		&Y^{v_i}_k = \hat{A}\tra{Y^{v_i}}{k-\tini}{k-1} + \hat{B}\tra{U^{v_i}}{k-\tini}{k-1},\\
		&\tra{(U,Y)^{v_i}}{1-\tini}{i} = 0,\quad Y^{v_i}_{i+1} = V_i-\mean[V]
	\end{split}
\end{equation}
\end{subequations}
with $\mean[Y_k^{v_i}]=\mean[U_k^{v_i}]=0$, $k\in\N$.
Note that due to the inherent causality requirement, i.e., the input $U_k$ and output $Y_k$ may only depend on past disturbances $V_s$, $s<k$, we have $\tra{(U,Y)^{v_i}}{1-\tini}{i} = 0$, $i\in\N$. Thus,
\[
	Y^{v_i}_{i+1} = \hat{A}\cdot 0+ \hat{B}\cdot 0 + V_i-\mean[V]=V_i-\mean[V].
\]
Splitting the LTI system~\eqref{eq:VARX} into decoupled subsystems~\eqref{eq:Superposition}, the additive disturbance~$V_i$, $i\in\N$ is converted into the initial output condition $Y_{i+1}^{v_i}=V_i-\mean[V]$. That is, the subsystems~\eqref{eq:Superposition} are disturbance free except for $\mean[V]$ in~\eqref{eq:SubNom}. A crucial insight of the above considerations is that Lemma~\ref{lem:StochFundam} and Corollary~\ref{coro:StochFundamPCE} can be applied to predict the future trajectories of~\eqref{eq:Superposition} without the need for Hankel matrices involving disturbance data.

\subsection{PCE basis structure} \label{sec:Structure}
To obtain a PCE reformulation of subsystems~\eqref{eq:Superposition} for further analysis, one needs to construct a joint PCE basis applicable to the inputs and outputs over the whole prediction horizon, i.e. $U_k$ and $Y_k$ for $k\in\I_{[1,N]}$.

Let $\Psi^{v}_k\coloneqq \{ \psi^n(\xi_k)\}_{n=0}^{L_v-1}$, $L_v\in\N^\infty$ be the finite-dimensional basis for $V_k$, $k\in\I_{[0,N-1]}$. Then, the i.i.d. disturbances $V_{k}$, $k\in \I_{[0,N-1]}$ admit exact PCEs 
\begin{equation} \label{eq:WkPCE}
		V_k = \sum_{n=0}^{L_v-1} \pce{v}^n \psi^n(\xi_k).
\end{equation}
As the disturbances are identically distributed, the PCEs of $V_k$, $k\in\I_{[0,N-1]}$ have the same algebraic structure of the basis functions $\psi^n$ and coefficients $\pce{v}^n$, $n\in\I_{[0,L_v-1]}$. The independence of $V_k$, $k\in\I_{[0,N-1]}$ is modelled by the use of different stochastic arguments $\xi_k$  in the PCE basis functions.
We construct the joint disturbance basis $\Phi\coloneqq \cup_{k=0}^{N-1} \Psi^{v}_k$ as
\begin{equation}\label{eq:BasisElement}
	\begin{split} 
		\Phi = &\Big\{ 1, \underbrace{\psi^1(\xi_0),...,\psi^{L_v-1}(\xi_0)}_{\Phi^{v_0}\setminus \{\psi^0(\xi_0)\}},
		\\ &\qquad...,\underbrace{\psi^1(\xi_{N-1}),...,\psi^{L_v-1}(\xi_{N-1})}_{\Phi^{v_{N-1}}\setminus \{\psi^0(\xi_{N-1})\}}\Big\},
	\end{split}
\end{equation}
which contains a total of $L=1+N(L_v-1)$ terms. We also enumerate the joint basis $\Phi$ from 0 to $L-1$ according to the sequence given above. A detailed description of the basis structure is omitted here for brevity. It can be found, e.g., in Section~2.3 of \citep{ou25polynomial}.

Recall that the PCEs of i.i.d. $V_k$, $k\in\I_{[0,N-1]}$ in the joint basis~$\Phi$ be $V_k=\sum_{j=0}^{L-1}\pce{v}_k^j\phi^j$. From the structure of~\eqref{eq:BasisElement} we observe that the basis functions related to the disturbance $V_k$, $k\in\I_{[0,N-1]}$ are $\phi^j$, $j\in \{0\}\cup\ik{k}$ with
\[
	\ik{k}\coloneqq \I_{[1+k(L_v-1),(k+1)(L_v-1)]}.
\] 
Conversely, given the PCE dimension $j$, using
\begin{equation*} \label{eq:k}
	k^\prime(j) \coloneqq \begin{cases}
		0,~&\text{for } j=0,\\
		k \text{ such that } j\in\ik{k}, \quad &\text{for } j\in\I_{[1,L-1]}
	\end{cases},
\end{equation*}
we see that the PCE coefficients $(\pce{u},\pce{y})^j$, $j\in\I_{[1,L-1]}$  relate to the disturbance $V_{k^\prime(j)}$. Note that $(\pce{u},\pce{y})^0$ is the expected value and thus it is linked to system~\eqref{eq:SubNom}.

With the index set $\ik{k}$,  $k\in\I_{[0,N-1]}$, we can identify the PCE coefficients corresponding to the disturbance $V_k$ in the joint basis $\Phi$
\begin{equation*}
	\pce{v}_k^j = \begin{cases}
		\pce{v}^0=\mean[V],~&\text{for } j=0\\
		\pce{v}^{j-k(L_v-1)},~&\text{for } j\in\ik{k}\\
		0,&\text{otherwise}
	\end{cases},
\end{equation*}
where $\pce{v}^{n}$, $n\in\I_{[0,L_v-1]}$ are defined in~\eqref{eq:WkPCE}. Moreover,  \citep[Proposition~1]{pan23stochastic} shows that all $Z_k$, $Z\in\{U,Y\}$ of~\eqref{eq:VARX} admit exact PCEs in this joint basis $\Phi$ for all $k \in \I_{[0,N]}$ over the entire horizon $N$. Therein the basis $\Phi$ entails the same basis directions but is indexed differently.

\subsection{Causality in PCE}
Given the joint basis $\Phi$, the PCEs of $\mean[Z]$ and $Z^{v_i}$, $Z\in\{U,Y\}$ with non-zero coefficients read
\begin{equation} \label{eq:RVDecompositionPCE}
	\mean[Z_k] = \pce{z}_k^0\phi^0,\quad Z_k^{v_i} = \sum_{j\in\ik{i}}\pce{z}_k^j\phi^j.
\end{equation}
Expressing the causality conditions $\tra{(U,Y)^{v_i}}{1-\tini}{i} = 0$ and $Y^{v_i}_{i+1} = V_i-\mean[V]$ in the PCE framework, we have
\begin{equation} \label{eq:Causality}
	\tra{(\pce{u},\pce{y})^j}{1-\tini}{k^{\prime}(j)}=0,~ \pce{y}_{k^{\prime}(j)+1}^j = \pce{v}^{I(j)}
\end{equation}
for all $j\in\I_{[1,L-1]}$, where $I(j)\coloneqq j-k^{\prime}(j)(L_v-1)$ \citep{ou25polynomial}. This way, the PCE coefficients of disturbances in~\eqref{eq:DynPCE} become the initial value $\pce{y}_{k^{\prime}(j)+1}^j = \pce{v}^{I(j)}$. It is straightforward to rewrite the dynamics of the PCE coefficients for $j=0$,
\begin{subequations} \label{eq:SuperpositionPCE}
	\begin{equation}\label{eq:VARXSimExp}
		\begin{split}
			&\pce{y}_k^0 = \hat{A}\tra{\pce{y}^0}{k-\tini}{k-1} + \hat{B}\tra{\pce{u}^0}{k-\tini}{k-1}+\mean[V],\\
			&\tra{(\pce{u},\pce{y})^0}{1-\tini}{0} = \tra{(\tilde{u},\tilde{y})}{1-\tini}{0},
		\end{split}
	\end{equation}
which corresponds to the nominal deterministic system~\eqref{eq:SubNom}. For $j\in\ik{i}$, $i\in\I_{[0,N-1]}$, we obtain
	\begin{equation} \label{eq:VARXSimDist}
			\pce{y}_k^j = \hat{A}\tra{\pce{y}^j}{k-\tini}{k-1} {+}\hat{B}\tra{\pce{u}^j}{k-\tini}{k-1}~\text{with}~\eqref{eq:Causality},
	\end{equation}
\end{subequations}
which corresponds to the stochastic error system~\eqref{eq:SubError} for $V_i$.
Thus, one can apply the fundamental lemma, cf.   \citep[Theorem~1]{willems05note}, to predict the future trajectory~$\tra{(\pce{u},\pce{y})}{1}{N}$ for all $j\in\I_{[0,L-1]}$.

\subsection{Propagation without past disturbance data in Hankel matrices}
We move on to show how one can forward propagate the PCE coefficients of LTI system~\eqref{eq:VARXPCE}, in a data-driven fashion without past disturbance data in the online computation. The crucial observation is that the reformulated dynamics~\eqref{eq:SuperpositionPCE} are deterministic systems---disturbed by the constant $\mean[V]$ in case of \eqref{eq:VARXSimExp} and undisturbed in case of \eqref{eq:VARXSimDist}.
Hence, we require a recording of an input-output trajectory $\traud{(u,y)}{1}{T}$ of the undisturbed system
\begin{subequations} \label{eq:VARXFree}
	\begin{align}
		&y_k^{\ud} = \hat{A}\traud{y}{k-\tini}{k-1} + \hat{B}\traud{u}{k-\tini}{k-1}, \label{eq:VARXDyn}\\
		&\traud{(u,y)}{1-\tini}{0} = \tra{(\tilde{u},\tilde{y})}{1-\tini}{0} \label{eq:VARXFreeIni}
	\end{align}
\end{subequations}
in the offline phase, where the superscript $\cdot^{\ud}$ denotes the undisturbed system and data. It is worth mentioning that the input ${u}^{\ud}$ of~\eqref{eq:VARXFree} is indeed identical to that of the disturbed system~\eqref{eq:VARXReal}. The superscript $\cdot^{\ud}$ of the input is used to maintain consistent notation with the undisturbed output $y^{\ud}$.

To obtain the data of~\eqref{eq:VARXFree}, we record the data of system~\eqref{eq:VARX} when it is temporarily undisturbed, i.e. $v=0$. However, for any continuously distributed disturbance, the disturbance realization is almost surely not identically zero over the past time instants. Therefore, the disturbance realizations in the past still affect the current and the future outputs.
Consider disturbances that act on system~\eqref{eq:VARXFree} at time step $k<0$, i.e. $v_k\neq0$ for $k<0$. Then we decompose the output as $y^{\ud} = \bar{y} + y^v$ with
	\begin{alignat*}{2} 
		&\bar{y}_k = \hat{A}\tra{\bar{y}}{k-\tini}{k-1} + \hat{B}\traud{u}{k-\tini}{k-1},\\
		&\tra{\bar{y}}{1-\tini}{0} = \tra{(\tilde{y}-y^v)}{1-\tini}{0},\\
		&y^v_k = \hat{A}\tra{y^v}{k-\tini}{k-1} + \hat{B}\cdot 0 + v_{k-1},\quad &&k\leq0\\
		& y^v_k = \hat{A}\tra{y^v}{k-\tini}{k-1} + \hat{B}\cdot 0 + 0,  && k\geq 1,
	\end{alignat*}
where $\bar{y}$ denotes the part of output unaffected by disturbances, and $y^v$ denotes the part affected by disturbances. Given a persistently exciting input sequence $\traud{u}{1}{T}$, the condition
\[\rank\left( \begin{bmatrix} 
	\Hankel_{\tini}(\traud{u}{1}{T})\\ 
	\Hankel_{\tini}(\tra{\bar{y}}{1}{T}) 
\end{bmatrix} \right) = \tini\cdot n_u + n_x\]
holds \citep{willems05note}. From the dynamics $y^v_k = \hat{A}\tra{y^v}{k-\tini}{k-1} + v_{k-1}$ for $k\leq 0$, we see that the elements of the vector $\tra{y^v}{1-\tini}{0}$ are affected by the randomly sampled realizations of disturbances $V_{[-\tini,-1]}$. It is then straightforward to show that $\rank\left(\Hankel_{\tini}(\tra{y^v}{1}{T})\right) = \tini\cdot\dimy$ holds under the mild condition $\rank(\hat{A}_1)=\dimy$, where $\hat{A}_1\in\R^{\dimy\times\dimy}$ is the first block of  $\hat{A}=[\hat{A}_1~|~\hat{A}_2~|~\cdots~|~\hat{A}_{\tini}]$. Thus, from $\Hankel_{\tini}(\traud{y}{1}{T}) = \Hankel_{\tini}\big(\tra{\bar{y}}{1}{T}) +  \tra{y^v}{1}{T}\big)$ we have the following rank condition
\begin{align*}
	&\rank \left( \begin{bmatrix}
		\Hankel_{\tini}(\traud{u}{1}{T})\\
		\Hankel_{\tini}(\traud{y}{1}{T})
	\end{bmatrix}\right) 
	=\rank\left( \begin{bmatrix*}[l]
		\Hankel_{\tini}(\traud{u}{1}{T})\\ 
		\Hankel_{\tini}(\tra{y^v}{1}{T})
	\end{bmatrix*} \right)\\
	=& \tini\cdot(\dimu+\dimy).
\end{align*}
Hence, the inclusion follows as
\[
\mathrm{colsp}\left( \begin{bmatrix} \Hankel_{\tini}(\traud{u}{1}{T})\\ \Hankel_{\tini}(\tra{\bar{y}}{1}{T}) \end{bmatrix}\right)\subseteq \mathrm{colsp}\left( \begin{bmatrix} \Hankel_{\tini}(\traud{u}{1}{T})\\ \Hankel_{\tini}(\traud{y}{1}{T}) \end{bmatrix}\right)
\]
since $\dimx\leq\tini\cdot\dimy$. The column spaces in the above inclusion are equivalent when $\dimx=\tini\cdot\dimy$. Therefore, given an arbitrary trajectory $\tra{(\tilde{u},\tilde{y})}{1-\tini}{0}$ of disturbed system~\eqref{eq:VARX}, there exists $g\in\R^{T-\tini+1}$ such that
\[
\begin{bmatrix}
	\Hankel_{\tini}(\traud{u}{1}{T-\tini})\\
	\Hankel_{\tini}(\traud{y}{1}{T-\tini})
\end{bmatrix} g = 
\begin{bmatrix}
	\tra{\tilde{u}}{1-\tini}{0}\\
	\tra{\tilde{y}}{1-\tini}{0} 
\end{bmatrix},
\]
while such a vector $g$ may not exist if the data $\tra{(u^{\ud},\bar{y})}{1}{T}$ is used as the stacked Hankel matrices  $\begin{bmatrix} 
		\Hankel_{\tini}(\traud{u}{1}{T})\\ 
		\Hankel_{\tini}(\tra{\bar{y}}{1}{T}) 
	\end{bmatrix}$ are not of full rank.

\begin{rem}[Obtaining undisturbed system data]
	There are different approaches to obtain the data of system~\eqref{eq:VARXFree}: 
		\begin{enumerate}[label=(\roman*)]
			\item Measure the input-output data of system~\eqref{eq:VARX} during periods when disturbances are absent or negligible. For instance, \citet{hong16probabilistic} suggest that the electricity demand in power systems exhibit reduced variability at midnight when the weather conditions are smooth and human activities are limited. Moreover, by appropriate input design, the effect of disturbances can be reduced such that the resulting input–output data can be treated as approximately noise-free \citep{ljung99system, forssell99time}.
			\item\label{itm:measurement} Deploy additional high-fidelity sensors to measure the realizations of disturbances of system~\eqref{eq:VARX} in the offline phase, e.g., environment temperature for building climate control and user demands in power systems. Then, based on the collected data, one can compute the undisturbed trajectory of~\eqref{eq:VARXFree}, cf.~Section~\ref{sec:Estimation}. This assumption on the availability of additional data, e.g., disturbance or state measurements, appears frequently in recent literature \citep{depersis20formulas, doerfler23on, alanwar23data,wolff24robust,disaro24equivalence}.
			\item Collect the input-output data of the disturbed system~\eqref{eq:VARX} and estimate the corresponding disturbance realizations. \citet{pan23stochastic} propose a least-squares estimator
			    \begin{equation} \label{eq:Estimator}
		    		\Hankel_1(\trad{\hat{v}}{0}{T-1}) = \Hankel_1(\trad{y}{1}{T}-\mean[V])(I-S^\dagger S)
    			\end{equation}
			with $S=\begin{bmatrix} \Hankel_{\tini}(\trad{y}{1-\tini}{T-1}) \\ \Hankel_{\tini}(\trad{u}{1-\tini}{T-1}) \end{bmatrix}$. Then, one can compute the undisturbed trajectory of~\eqref{eq:VARXFree} as in Approach~\ref{itm:measurement}. Additionally, \citet{turan22data} develop a state observer for systems with unknown disturbances.
		\end{enumerate}
\end{rem}
Using the data of~\eqref{eq:VARXFree} to construct Hankel matrices, one may wonder whether the predicted trajectory $\tra{(u,y)}{1}{N}$ is guaranteed to satisfy the  dynamics~\eqref{eq:VARXDyn} since the effects of  past disturbances before data collection persist. Intuitively, the answer is positive as the dynamics~\eqref{eq:VARXDyn} still hold for $\traud{(u,y)}{1}{T}$ during data collection.

It should be noted that one cannot directly apply the original stochastic fundamental lemma, i.e. Lemma~\ref{lem:StochFundam}, to the undisturbed system~\eqref{eq:VARXFree}. In this case, $v=0$ and $\Hankel_N(\trad{v}{0}{T-1})=0$ hold. Therefore, $(u,0)_{[0,T-1]}$ is not persistently exciting. Put differently, in case of undisturbed input-output measurements, our previous result allows only to predict undisturbed future trajectories.

\begin{assum}[Persistently exciting data] \label{ass:PE}
	We assume that the collected data $\traud{(u,y)}{1}{T}$ satisfy
	\[
		\rank\left( \begin{bmatrix*}[l]
			\Hankel_{\tini+N}(\traud{u}{1}{T}) \\ \Hankel_{\tini}(\traud{y}{1}{T-N+1})
		\end{bmatrix*}\right) = (\tini+N)n_u + \tini n_y,
	\]
	i.e., the collected data $\traud{(u,y)}{1}{T}$ are persistently exciting.
\end{assum}
The next result shows how to capture the undisturbed realization dynamics~\eqref{eq:VARXDyn}.
\begin{cor} \label{cor:Prediction}
	 Let Assumption~\ref{ass:Sys}, \ref{ass:iid}, and \ref{ass:PE} hold. Then a sequence $\tra{(u,y)}{1-\tini}{N}$ with an arbitrary initial condition $\tra{(u,y)}{1-\tini}{0}=\tra{(\tilde{u},\tilde{y})}{1-\tini}{0}$ satisfies dynamics~\eqref{eq:VARXFree}
	  if and only if there exists $g\in\R^{T-N-\tini+1}$ such that
\begin{equation} \label{eq:FundaLemma}
	\left[\begin{array}{ll} \Hankel_{\tini+N}(\traud{u}{1}{T})\\ \Hankel_{\tini+N}(\traud{y}{1}{T}) \end{array}\right] g = 
	\left[\begin{array}{ll} \tra{u}{1-\tini}{N} \\ \tra{y}{1-\tini}{N} \end{array}\right].
\end{equation}
\end{cor}
Corollary~\ref{cor:Prediction} is a straight-forward extension of the deterministic fundamental lemma by \citet{willems05note} towards disturbed initial conditions
and the proof follows along the same lines as Theorem~1 of \citep{berberich22linear}. 
Specifically, Corollary~\ref{cor:Prediction} allows predicting the future trajectory $\tra{(u,y)}{1}{N}$ along dynamics~\eqref{eq:VARXDyn} even when the initial trajectory $\tra{(\tilde{u},\tilde{y})}{1-\tini}{0}$ is disturbed, whereas the deterministic fundamental lemma requires an undisturbed initial trajectory.
Then we arrive at our main results, which are provided in Lemma~\ref{lem:j0}, \ref{lem:jOther}, and \ref{lem:PropRV} below.

\begin{lem}[Propagation in expectation ($j=0$)] \label{lem:j0}
	Consider system~\eqref{eq:Dyn} and a trajectory $\traud{(u,y)}{1}{T}$ of~\eqref{eq:VARXFree}. Let Assumptions~\ref{ass:Sys}, \ref{ass:iid}, and \ref{ass:PE} hold. Then given a measured initial condition $\tra{(\pce{u},\pce{y})^0}{1-\tini}{0}$ = $\tra{(\tilde{u},\tilde{y})}{1-\tini}{0}$,
	$\tra{(\pce{u},\pce{y})^0}{1}{N}$ is a trajectory of~\eqref{eq:DynPCE} for $j=0$ if and only if there exists $\pce{g}^0\in\R^{T-N-\tini+1}$ such that
	\begin{subequations} \label{eq:UYj0}
	\begin{equation} \label{eq:TrajUYr}
		\left[\begin{array}{ll} \Hankel_{\tini}(\traud{u}{1}{T-\tini})\\ \Hankel_{\tini}(\traud{y}{1}{T-\tini}) \\ \midrule \Hankel_{N}(\traud{u}{\tini+1}{T}) \\ \Hankel_{N}(\traud{y}{\tini+1}{T})  \end{array}\right] \pce{g}^0 = 
		\left[\begin{array}{ll} \tra{\tilde{u}}{1-\tini}{0} \\ \tra{\tilde{y}}{1-\tini}{0} \\ \midrule \tra{\pce{u}^0}{1}{N} \\ \tra{y^u}{1}{N} \end{array}\right]
	\end{equation}
	holds, where $ y_k^u = \pce{y}_k^0 - \sum_{i=1}^{k}y_i^{v}$, $y_1^v=\mean[V]$ and the system output $\tra{y^v}{2}{N}$ is given by
	\begin{multline} \label{eq:yw}
		\tra{y^v}{2}{N} = \Hankel_{N-1}(\traud{y}{\tini+1}{T})\cdot\\
		\begin{bmatrix*}[l] \Hankel_{\tini+N-1}(\traud{u}{1}{T}) \\ \Hankel_{\tini-1}(\traud{y}{1}{T-N})\\ \Hankel_{1}(\traud{y}{\tini}{T-N+1}) \end{bmatrix*}^\dagger \begin{bmatrix*}[l]  0_{(\tini+N-1)\dimu\times1} \\ 0_{(\tini-1)\dimy\times1} \\ \mean[V] \end{bmatrix*}.
	\end{multline}
	\end{subequations}
\end{lem}
\begin{proof}
    The core idea of the proof is to decompose $\pce{y}_k^0$ into two parts $\pce{y}_k^0=y_k^u + \sum_{i=1}^{k}y_i^v$, where $y^u$ related to the input $\pce{u}^0$ and $y^v$ related to the disturbances.
	
	For $j=0$ it holds that $\pce{v}_k^0=\mean[V]$, $k\in\I_{[0,N-1]}$. Then we split system~\eqref{eq:VARXPCE} into a disturbance free system
	\begin{subequations}
		\begin{align}
			&y_k^u = \hat{A}\tra{y^u}{k-\tini}{k-1} + \hat{B}\tra{\pce{u}^0}{k-\tini}{k-1}, \label{eq:DynR}\\
			&\tra{(y^u,\pce{u}^0)}{1-\tini}{0} = \tra{(\tilde{y},\tilde{u})}{1-\tini}{0},
		\end{align}
	and an error system
		\begin{align}
			&y_k^{\tilde{v}} = \hat{A}\tra{y^{\tilde{v}}}{k-\tini}{k-1} + \mean[V],  \label{eq:DynE}\\
			&\tra{y^{\tilde{v}}}{1-\tini}{0} = 0.
		\end{align}
	\end{subequations}
	Moreover, \eqref{eq:DynE}  can be rewritten as $y^{\tilde{v}}_k = \sum_{i=1}^{k}y^v_i$, $k\in\I_{[1,N]}$ with the autonomous system
	\begin{subequations}
			\begin{align}
					&y^v_k = \hat{A}\tra{y^v}{k-\tini}{k-1}+\hat{B}\cdot0,\label{eq:DynTE}\\
					&\tra{y^v}{1-\tini}{0} = 0,\quad y_1^v =\mean[V].
				\end{align}
		\end{subequations}
	We notice that $y^v_1 \neq \hat{A}\tra{y^v}{1-\tini}{0}=0$. Therefore, we choose $\tra{y^v}{2-\tini}{1}=[0_{1\times(\tini-1)\dimy},\mean[V]^\top]^\top$ to be the initial condition, which is indeed not a trajectory of~\eqref{eq:DynTE}. However, applying Corollary~\ref{cor:Prediction} we can compute $\tra{y^v}{1}{N}$. That is, $\tra{y^v}{1}{N}$ satisfies \eqref{eq:DynTE} if and only if there exist a $g^v\in\R^{T-N-\tini+2}$ such that  
	\begin{equation} \label{eq:yv_constraint}
	 \left[\begin{array}{ll} \Hankel_{\tini+N-1}(\traud{u}{1}{T}) \\ \midrule \Hankel_{\tini-1}(\traud{y}{1}{T-N})\\ \Hankel_{1}(\traud{y}{\tini}{T-N+1}) \\ \Hankel_{N-1}(\traud{y}{\tini+1}{T}) \end{array}\right] g^v =
	 \left[\begin{array}{ll}  0_{(\tini+N-1)\dimu\times1} \\ \midrule 0_{(\tini-1)\dimy\times1} \\ \mean[V] \\ \tra{y^v}{2}{N} \end{array}\right].
	\end{equation}
	From the above equation we can explicitly compute the trajectory $\tra{y^v}{2}{N}$ as~\eqref{eq:yw}.
	Again, applying Corollary~\ref{cor:Prediction} to dynamics~\eqref{eq:DynR}, we conclude \eqref{eq:TrajUYr}. From $\pce{y}_k^0=y_k^u+y_k^{\tilde{v}}=y_k^u+\sum_{i=1}^{k}y_i^v$, we have the predicted trajectory $\tra{(\pce{u},\pce{y})^0}{1}{N}$.\End
\end{proof}
Note that \eqref{eq:yw} is given to represent $\tra{y^v}{2}{N}$ in closed form. Indeed, $\tra{y^v}{2}{N}$ can be obtained via solving the optimization problem
\[
	\min_{g^v}~1 \quad\text{s.t.}\quad \eqref{eq:yv_constraint}
\]
without the explicit computation of Moore-Penrose inverse in~\eqref{eq:yw}. Moreover, since the distribution of $V$ and hence $\mean[V]$ are known, one can compute the forward propagation of $\mean[V]$, i.e. the trajectory $\tra{y^v}{1}{N}$, in advance.

\begin{lem}[Propagation for $j\in\I_{[1,L-1]}$]\label{lem:jOther}
	Let the conditions of Lemma~\ref{lem:j0} hold.
	Then, for all $j \in \I_{[1,L-1]}$, $\tra{(\pce{u},\pce{y})^j}{1}{N}$ is a trajectory of \eqref{eq:DynPCE} if and only if there exist $\pce{g}^j \in \R^{T-N-\tini+1}$ such that
	\begin{subequations}\label{eq:UYjRest}
		\begin{gather}
			\tra{(\pce{u},\pce{y})^j}{1}{k^{\prime}(j)}=0,\quad \pce{y}_{k^\prime(j)+1}^j = \pce{v}^{I(j)}, \label{eq:UY0}\\
			\left[\begin{array}{ll} \Hankel_{\tini-1}(\traud{u}{1}{T-N-1})\\ \Hankel_{\tini-1}(\traud{y}{1}{T-N-1}) \\ \midrule \Hankel_{\bar{N}}(\traud{u}{\tini}{T-k^{\prime}-1}) \\ \Hankel_{\bar{N}}(\traud{y}{\tini}{T-k^{\prime}-1}) \end{array}\right] \pce{g}^j 
			=\left[\begin{array}{ll} 0_{(\tini-1)\dimu\times1} \\ 0_{(\tini-1)\dimy\times1} \\ \midrule\tra{\pce{u}^j}{k^{\prime}(j)+1}{N} \\ \tra{\pce{y}^j}{k^{\prime}(j)+1}{N} \end{array}\right], \label{eq:PCEUYj}
		\end{gather}
	\end{subequations}
	where $\bar{N}=N-k^{\prime}(j)$.
\end{lem}
\begin{proof}
	The causality condition in the PCE framework, i.e. \eqref{eq:Causality}, implies \eqref{eq:UY0}.
	Similar to Lemma~\ref{lem:j0}, we notice that
	\[
		\pce{y}_{k^{\prime}(j)+1}^j\neq \hat{A}\tra{\pce{y}^j}{k^{\prime}(j)-\tini+1}{k^{\prime}} +  \hat{B}\tra{\pce{u}^j}{k^{\prime}(j)-\tini+1}{k^{\prime}(j)}
	\]
	with $\tra{(\pce{u},\pce{y})^j}{k^{\prime}(j)-\tini+1}{k^{\prime}(j)}=0$. Hence, the initial condition is given as $\tra{(\pce{u},\pce{y})^j}{k^{\prime}(j)-\tini+2}{k^{\prime}(j)}=0$ and $\pce{y}_{k^{\prime}(j)+1}^j=\pce{v}^{I(j)}$, while $\pce{u}_{k^{\prime}(j)+1}^j$ remains an input variable. Then applying Corollary~\ref{cor:Prediction}, \eqref{eq:PCEUYj} immediately follows.\End
\end{proof}
As a by-product, the prediction horizon of \eqref{eq:DynPCE} for $j\in\I_{[1,L-1]}$ is shortened from $N$ to $N-k^{\prime}(j)$. Consequently, we have smaller Hankel matrices and less PCE coefficients as decision variables when we apply Lemma~\ref{lem:jOther} to stochastic OCPs. Therefore, Lemma~\ref{lem:jOther} also accelerates the computation in numerical implementations, see the numerical example in Section~\ref{sec:Simulation}.
Summarizing Lemmas~\ref{lem:j0}-\ref{lem:jOther}, we conclude the following lemma in random variables.

\begin{lem}[Propagation in random variables]\label{lem:PropRV}
	Let Assumptions~\ref{ass:Sys}, \ref{ass:iid}, and \ref{ass:PE} hold.
	Then $\tra{(U,Y)}{1}{N}$ is a trajectory of \eqref{eq:Dyn} for a measured initial trajectory $\tra{(\tilde{u},\tilde{y})}{1-\tini}{0}$ if and only if there exist $g\in\R^{T-N-\tini+1}$ and $G^{v_i}\in\splx{T-N-\tini+1}$, $i\in\I_{[0,N-1]}$ such that $Z_k = \mean[Z] + \sum_{i=0}^{k-1} Z_k^{v_i}$, $Z\in\{U,Y\}$, where
	\begin{enumerate}[label=(\roman*)]
		\item\label{Cond1} $\mean[\tra{(U,Y)}{1}{N}]=\tra{(\pce{u},\pce{y})^0}{1}{N}$ and $g=\pce{g}^0$ in Lemma~\ref{lem:j0},
		\item\label{Cond2}  $\tra{(U,Y)^{v_i}}{1-\tini}{i} = 0$, $Y^{v_i}_{i+1} = V_i-\mean[V]$, and
		\[
		\left[\begin{array}{ll} \Hankel_{\tini-1}(\traud{u}{1}{T-N-1})\\ \Hankel_{\tini-1}(\traud{y}{1}{T-N-1}) \\ \midrule \Hankel_{N-i}(\traud{u}{\tini}{T-i-1}) \\ \Hankel_{N-i}(\traud{y}{\tini}{T-i-1}) \end{array}\right] G^{v_i} 
		=\left[\begin{array}{ll} 0_{(\tini-1)\dimu\times1} \\ 0_{(\tini-1)\dimy\times1} \\ \midrule \tra{U^{v_i}}{i+1}{N} \\ \tra{Y^{v_i}}{i+1}{N} \end{array}\right].
		\]
	\end{enumerate}
\end{lem}
\begin{proof}
	First we prove that the conditions~\ref{Cond1}-\ref{Cond2} are necessary. Consider any trajectory~$\tra{(U,Y)}{1}{N}$ of system~\eqref{eq:Dyn} and its decomposition~\eqref{eq:Superposition}. Replacing all the inputs and outputs with their PCEs as~\eqref{eq:RVDecompositionPCE}, we obtain the PCE reformulated dynamics~\eqref{eq:SuperpositionPCE}, which are equivalent to the data-driven representations~\eqref{eq:UYj0} and \eqref{eq:UYjRest} as Lemmas~\ref{lem:j0}-\ref{lem:jOther} have shown. Since $\phi^0=1$, the condition~\ref{Cond1} directly follows from \eqref{eq:UYj0}. Then for the PCE representation~\eqref{eq:UYjRest}, $j\in\ik{i}$, we multiply them with the corresponding PCE basis functions $\phi^j$, $j\in\ik{i}$ and sum the results over. Let $G^{v_i}\coloneqq \sum_{j\in\ik{i}}\pce{g}^j\phi^j$, we get the condition~\ref{Cond2}.
	
	Next we show the sufficiency of the conditions~\ref{Cond1}-\ref{Cond2}. The stochastic fundamental lemma, i.e. Lemma~\ref{lem:StochFundam}, indicates that the trajectories~$\tra{(\mean[U],\mean[Y])}{1}{N}$ and $\tra{(U,Y)^{v_i}}{1}{N}$ are trajectories of the decomposed systems~\eqref{eq:SubNom} and \eqref{eq:SubError}, respectively, when the conditions \ref{Cond1}-\ref{Cond2} hold. Thus, $(U,Y)_k = (\mean[U],\mean[Y])_k + \sum_{i=0}^{k-1} (U,Y)_k^{v_i}$, $k\in\I_{[1,N]}$ satisfy the dynamics~\eqref{eq:VARX} as well as~\eqref{eq:Dyn}.\End
\end{proof}
\begin{rem}[Extension to $\tra{(\tilde{U},\tilde{Y})}{1-\tini}{0}$] \label{rem:UncertainIni}
	Consider system~\eqref{eq:Dyn} with an uncertain initial input-output trajectory $\tra{(U,Y)}{1-\tini}{0}=\tra{(\tilde{U},\tilde{Y})}{1-\tini}{0}$. One can split $Z_k = \mean[Z] +Z^{\ini} + \sum_{i=0}^{k-1} Z_k^{v_i}$, $Z\in\{U,Y\}$ with
	\begin{align*}
		&Y^{\ini}_k = \hat{A}\tra{Y^{\ini}}{k-\tini}{k-1} + \hat{B}\tra{U^{\ini}}{k-\tini}{k-1},\\
		&\tra{(U,Y)^{ini}}{1-\tini}{i} = \tra{(\tilde{U}-\mean[\tilde{U}],\tilde{Y}-\mean[\tilde{Y}])}{1-\tini}{0}.
	\end{align*}
	The forward propagation for the PCE coefficients of $Z^{\ini}$, $Z\in\{U,Y\}$ is a simplified case of the propagation of expectation in Lemma~\ref{lem:j0} with $\mean[V]=0$, i.e., the predicted trajectory of the PCE coefficients of an uncertain initial condition satisfies~\eqref{eq:TrajUYr}.
	Then, besides the conditions of Lemma~\ref{lem:PropRV}, $\tra{(U,Y)}{1}{N}$ is a trajectory of \eqref{eq:Dyn} if and only if there exists $G^{\ini}\in\splx{T-N-\tini+1}$ such that 
	\[
			\left[\begin{array}{ll} \Hankel_{\tini}(\traud{u}{1}{T-N-1})\\ \Hankel_{\tini}(\traud{y}{1}{T-N-1}) \\ \midrule \Hankel_{N}(\traud{u}{\tini}{T-i-1}) \\ \Hankel_{N}(\traud{y}{\tini}{T-i-1}) \end{array}\right] G^{\ini}
		=\left[\begin{array}{ll} \tra{(\tilde{U}-\mean[\tilde{U}])}{1-\tini}{0}\\ \tra{(\tilde{Y}-\mean[\tilde{Y}])}{1-\tini}{0} \\ \midrule \tra{U^{\ini}}{1}{N} \\ \tra{Y^{\ini}}{1}{N} \end{array}\right].
	\]
	A similar result in PCE coefficients also immediately follows and is omitted for brevity.
\end{rem}

\begin{table*}[t!]
	\caption{Comparison of different fundamental lemmas for LTI systems.}
	\label{tab:ComparisonLemmas}
	\centering
		\begin{tabular*}{\textwidth}{@{\extracolsep{\fill}} llcc}
			\toprule
			Lemma  & \hspace{30pt} System  & Data & Number of entries in $\Hankel$\\
			\midrule
			\citet{willems05note}  & \hspace{2pt}$\begin{array}{rcl} x_{k+1} &= & Ax_k+Bu_k,\\ y_k &=& Cx_k+Du_k\end{array}$ & $(u,y)^{\da}$ & $(\tini+N)(n_u+n_y)n_g$\\
			\citet{berberich22linear} & \hspace{2pt}$\begin{array}{rcl} x_{k+1} &= & Ax_k+Bu_k+e,\\ y_k &=& Cx_k+Du_k+r\end{array}$ & $(u,y)^{\da}$ & $(\tini+N)(n_u+n_y)n_g$ \\
			\citet{kerz23data}     & $\begin{array}{rcl} X_{k+1}&=&AX_k+BU_k+EW_k \\ \hat{X}_k &=&X_k+H_k \end{array} $ & $(u,x,w)^{\da}$ & $N(n_u+n_x+n_w)n_g$\\
			\citet{chiuso25harnessing} & $\begin{array}{rcl} \mean[Y_{k}] &= & \cdots,\\ \mathrm{Cov}[Y_k] &=& \cdots \end{array}$ & $(u,y)^{\da}$ & $(\tini+N)(n_u+n_y)n_g$\\  
			\begin{tabular}{@{}l@{}} \citet{pan23stochastic} \\ (Lemma~\ref{lem:StochFundam})\\ \end{tabular} &  $\begin{array}{rcl} X_{k+1} &= & AX_k+BU_k+EW_k,\\ Y_k &=& CX_k+DU_k\end{array}$ & $(u,y,w)^{\da}$ & $(\tini+N)(n_u+n_y+n_w)n_gL$ \\
			Lemma~\ref{lem:PropRV} &
			\hspace{11pt}$\begin{array}{rcl} Y_k & = & \hat{A}\tra{Y}{k-\tini}{k-1} \\ & &+\hat{B}\tra{U}{k-\tini}{k-1} + V_{k-1}\end{array}$ 
			& $(u,y)^{\ud}$ & $\approx(\tini+\frac{N}{2})(n_u+n_y)n_gL$ \\ 
			\bottomrule
		\end{tabular*}
\end{table*}

In Table~\ref{tab:ComparisonLemmas}, we compare the different fundamental lemmas, including the stochastic fundamental lemma by \citet{pan23stochastic} (Lemma~\ref{lem:StochFundam}) and Lemma~\ref{lem:PropRV}; \citet{berberich22linear} consider deterministic affine systems, \citet{kerz23data} propose a pre-stabilized deterministic fundamental lemma for stochastic LTI systems subject to process disturbances and measurement noise, while \citet{chiuso25harnessing} propagates the first two moments propagation using a separation principle. We also compare the required data and the number of entries in the Hankel matrices, denoted by $\Hankel$ in Table~\ref{tab:ComparisonLemmas}, in different variants of the fundamental lemma. The number of entries in a Hankel matrix corresponds to its size and thus, to some extent, it reflects the computational complexity of applying the corresponding results. Here $n_g$ denotes the dimension of vector $g$ or $G$ and $L$ denotes the PCE dimension, which is proportional to the prediction horizon $N$ and is $L=1+N(L_w-1)$ for measured initial condition. Let $n_w=n_y$ and $\tini\ll N$, then we have the approximation
\[
	\frac{\text{number entries in~}\Hankel~\text{of Lemma~\ref{lem:PropRV}}}{\text{number entries in~}\Hankel~\text{of Lemma~\ref{lem:StochFundam}}}\approx\frac{n_u+n_y}{2(n_u+2n_y)}.
\]
The computational complexity of finding a feasible $G$ in Lemma~\ref{lem:StochFundam} is $\mcl{C}_1=\mcl{O}\Big((\tini(\dimu+\dimy)+(\tini+N)\dimy)\left(\tini(\dimu+\dimy)+(\tini+N)\dimy+L\right)n_g \Big)$ for the case $\dimw=\dimy$, while the corresponding computational complexity in Lemma~\ref{lem:PropRV} is $\mcl{C}_2=\mcl{O}(\tini(\dimu{+}\dimy)\cdot\left(\tini(\dimu{+}\dimy)+L\right)n_g)$. Comparing the computational complexities $\mcl{C}_1$ and $\mcl{C}_2$, one obtains
\begin{align} 
	&\frac{\mcl C_2}{\mcl C_1} {=} 
	\scalebox{0.95}{$\displaystyle\left(1 {+} \frac{(\tini{+}N)\dimy}{\tini(\dimu{+}\dimy)}\right)^{-1} \cdot \left(1{+}\frac{(\tini{+}N)\dimy}{\tini(\dimu{+}\dimy){+}N\dimy{+}1}\right)^{-1}$} \nonumber\\
	\leq & \scalebox{0.9}{$\displaystyle \left(1 {+} \frac{(\tini{+}N)\dimy}{(\tini{+} N)(\dimu{+}\dimy)}\right)^{-1} \cdot \left(1{+}\frac{(\tini{+}N)\dimy}{\tini(\dimu{+}\dimy){+}N\dimy{+}N\dimu}\right)^{-1}$} \nonumber\\
	= &\left(1 + \frac{\dimy}{\dimu+\dimy} \right)^{-2}. \label{eq:CC}
\end{align}
That is, the computational complexity of Lemma~\ref{lem:PropRV} is less than  $\left(1 + \frac{\dimy}{\dimu+\dimy} \right)^{-2}$ of that of Lemma~\ref{lem:StochFundam} for the same system and prediction horizon $N$.

\begin{rem}[Computational complexity] \label{rem:Computation}
Consider a system with $\dimy\geq\dimu$ and a prediction horizon $N\geq\tini\geq 1$. The maximum of $\frac{\mcl C_2}{\mcl C_1}$ is $\max \frac{\mcl C_2}{\mcl C_1} =\frac13$, which occurs when $\dimy=\dimu$ and $N=\tini$. In other words, Lemma~\ref{lem:PropRV} reduces the computational complexity by at least $66.7\%$ compared to Lemma~\ref{lem:StochFundam} when $\dimy=\dimu$.
\end{rem}

\subsection{Estimation of Disturbance-Free Data} \label{sec:Estimation}
In this section, we first propose a procedure to find a stabilizing feedback gain for the stochastic LTI system~\eqref{eq:VARX}. Then we show how one can estimate an undisturbed trajectory of~\eqref{eq:VARXFree} from the recorded realization trajectory $\trad{(u,y,v)}{1-\tini}{T}$ of~\eqref{eq:VARXReal}.

Given an unstable stochastic LTI system with a sequence of randomly sampled input in the offline data collection phase, the system response grows exponentially over time. Consequently, the constructed Hankel matrix of output data may exhibit ill-conditioning, which causes numerical issues. To prevent the system response from diverging, we design a stabilizing feedback controller
\[
	u_k^{\da}=Kz_k^{\da} +\epsilon_k^{\da} \text{ with }z_k^{\da} = \begin{bmatrix} \traud{u}{k-\tini}{k-1}\\ \traud{y}{k-\tini}{k-1} \end{bmatrix},
\]
where $\epsilon_k^{\da}$ is small additive random noise to guarantee persistency of excitation of the inputs. Based on recorded input-output data, one can compute $K$ via solving an optimization problem, see \citet{doerfler23on} for state feedback and \citet{pan25data} for output feedback. When disturbance measurements are unavailable, an estimator for the disturbance realizations is required to obtain $K$, e.g., the least-square estimator~\eqref{eq:Estimator}.
Note that the accuracy of the estimated disturbances remains an open question. Moreover, the recorded input-output trajectory and estimated disturbance realizations, i.e., $\trad{(u,y,\hat{v})}{0}{T-1}$, satisfy dynamics~\eqref{eq:VARX} but with different system matrices $\hat{A}$ and $\hat{B}$, cf. Corollary~3 by \citet{pan25data}. Therefore, Lemma~\ref{lem:StochFundam} and \ref{coro:StochFundamPCE} remain valid with the estimated disturbances.

Due to the estimation error of disturbance realizations, the computed feedback gain $K$ is not guaranteed to stabilize the system~\eqref{eq:DynReal}. Here we propose the following experimental procedure to resolve this issue:
\begin{enumerate}[label=(\roman*)]
	\item Sample an input-state trajectory of system~\eqref{eq:DynReal} for a short length, i.e. $\trad{(u,y)}{0}{T}$ with small $T$, and estimate the disturbance $\trad{\hat{v}}{0}{T-1}$.
	\item Solve the optimization problem proposed in \citet{pan25data} with the sampled data and compute the corresponding $K$.
	\item Let the feedback be $u_k = Kz_k + \epsilon_k$, where $\epsilon_k$ is uniformly sampled from a small interval, e.g., $[-10^{-3},10^{-3}]$, and implement the input.
	\item Sample the output data of system~\eqref{eq:DynReal} and check whether it stays in a neighborhoodf of the origin. If not, go to step (i) and repeat the procedure with current input policy; otherwise the procedure terminates.
\end{enumerate}
With a stabilizing feedback $K$, we can sample $\trad{(u,y)}{0}{T}$ of~\eqref{eq:DynReal} and construct the Hankel matrices with acceptable condition numbers. 

Consider system~\eqref{eq:VARX} and a corresponding input-output realization trajectory $\trad{(u,y)}{1-\tini}{T}$. We estimate the disturbance realizations $\trad{\hat{v}}{1-\tini}{T-1}$ for the least-square estimator.
Let $\trad{u}{0}{T-1}$ be persistently exciting of order $\hat{T}+\dimx$. Then an estimation of an undisturbed trajectory $\traud{(u,y)}{1}{\hat{T}}$ of system~\eqref{eq:VARXFree} can be computed from
\begin{equation}\label{eq:PredictorT}
\begin{bmatrix*}[l] \Hankel_{\hat{T}+\tini}(\trad{u}{1-\tini}{T}) \\ \Hankel_{\hat{T}+\tini}(\trad{y}{1-\tini}{T}) \\ \Hankel_{\hat{T}}(\trad{\hat{v}}{0}{T-1}) \end{bmatrix*} g = \begin{bmatrix*}[l] \traud{u}{1-\tini}{\hat{T}} \\ \traud{y}{1-\tini}{\hat{T}} \\ 0_{\hat{T}n_y\times1} \end{bmatrix*},
\end{equation}
where the initial trajectory $\traud{(u,y)}{1-\tini}{0}$ is an arbitrary piece of the recorded data $\trad{(u,y)}{1-\tini}{T}$ and is thus known. It is straightforward to see that $\hat{T}\ll T$. Thus, we may need to repeat the procedure \eqref{eq:PredictorT} until an undisturbed trajectory of sufficient length to construct Hankel matrices is obtained.
Moreover, similar to Lemma~3 by \citet{kerz23data}, we can modify~\eqref{eq:PredictorT} to compute a trajectory $\tra{(\epsilon^{\da},z^{\ud})}{1}{\hat{T}}$ close to the origin
\begin{equation} \label{eq:PredictorModify}
\begin{bmatrix*}[l] \Hankel_{\hat{T}}(\trad{(u-Kz)}{1}{T})  \\ \Hankel_{\hat{T}}(\trad{z}{1}{T}) \\ \Hankel_{\hat{T}}(\trad{\hat{v}}{0}{T-1})  \end{bmatrix*} g = \begin{bmatrix*}[l]
	\trad{\epsilon}{1}{\hat{T}} \\ \traud{z}{1}{\hat{T}} \\ 0_{\hat{T}n_y\times 1} \end{bmatrix*},
\end{equation}
where the input-output trajectory is included in $\traud{z}{1}{\hat{T}}$. This way, we construct the Hankel matrices from the computed data of~\eqref{eq:VARXFree} with acceptable condition numbers.

\subsection{Extensions Towards Different Perspectives} \label{sec:Extension}
The reader may wonder how to relax Assumption~\ref{ass:iid}, which assumes the disturbances to be i.i.d. with known distribution. First we consider non-i.i.d. disturbances with known distributions. In this case, one chooses a different PCE basis for each $V_k$, $k\in\I_{[0,N-1]}$ and then constructs a joint PCE basis similar to~\eqref{eq:BasisElement}. This way, one still obtains the PCE coefficients of the disturbances, i.e. $\pce{v}_k^j$ for $k\in\I_{[0,N-1]}$, $j\in\I_{[0,L-1]}$. Then, it is straightforward to show that Lemmas~\ref{lem:j0}, \ref{lem:jOther}, and~\ref{lem:PropRV} also hold for non-i.i.d. disturbances along the same line as the i.i.d. case. The detailed proofs are omitted here.

Moreover, in a distributionally robust setting, the requirement for knowledge of the disturbance distribution can be relaxed to only its first two moments. Consider the disturbance $V$ with covariance decomposition
\[
	\var[V]=\Sigma_v = M_vM^\top_v.
\]
Then we get $V=\mean[V] + M_v\xi_v$, where $\xi_v$ is the normalized $V$ with zero mean. The PCE of $V$ follows as $V = \sum_{n=0}^{\dimy} \pce{v}^n\psi^n(\xi_v)$ with
\[
	\pce{v}^0=\mean[V],\quad \begin{bmatrix}
	\pce{v}^1 & \pce{v}^2 &\cdots & \pce{v}^{\dimy}
	\end{bmatrix}  = M_v.
\]
As the distributions of $V$ and $\xi_v$ are unknown, the PCE basis $\psi^n(\xi_v)$ can be selected as an arbitrary orthonormal basis, which may deviate from the basis induced by the distribution of $V$. We refer to \citet{pan25data} for a detailed analysis of recursive feasibility and other closed-loop properties.

When it comes to the measurement noise, one may consider a stochastic LTI system 
\begin{subequations} \label{eq:dyn_measure}
	\begin{align}
		X_{k+1} &= AX_k + BU_k + EW_k,\quad X_0=\tilde{X}_0,\\
		Y_k &= CX_k + F H_k,
	\end{align}
\end{subequations}
where $H_k\in\splx{n_h}$ is the measurement noise. \citet{casals12from} have shown that \eqref{eq:dyn_measure} is also equivalent to VARX model~\eqref{eq:VARX}, where $V$ captures the effect of both process disturbances and measure noise. Therefore, one can predict $\tra{Y}{1}{N}$ that represents the future output measurements in random variables via the original and the modified stochastic fundamental lemma, rather than the actual output values. On the other hand, in terms of closed-loop control regularized data-driven schemes can also be employed and have shown good performance, see, e.g., \citep{doerfler23bridging,kladtke23implicit, breschi23data}.
\begin{exmp}[Measurement noise]
Consider VARX model~\eqref{eq:VARX} in the presence of both process disturbances and measurement noise. Let the output $Y$ be scalar. The process disturbances follow a Gaussian distribution $\mcl{N}(\mu,\sigma^2)$, and measurement noise follows a uniform distribution $\mcl{U}(-a,a)$ with $a>0$. Since the term $V_k$ summarizes the combined effect of process disturbances and measurement noise, the PCE of $V_k$ is given by
\begin{gather*}
	V_k = \mu \cdot \psi^0 + \sigma \cdot \psi^1(\xi_k^{p}) + 2a\cdot \psi^2(\xi_k^{m})\\
	\text{with}\quad \psi^0=1,\quad \psi^1(\xi_k^{p})= \xi_k^{p}, \quad \psi^2(\xi_k^{m}) = \xi_k^{m}-0.5,
\end{gather*}
where $\xi_k^{p}$ and $\xi_k^{m}$ are independent stochastic arguments and $\xi_k^{p}\sim\mcl{N}(0,1)$, $\xi_k^{m}\sim\mcl{U}(0,1)$. The additional PCE term $2a\cdot \psi^2(\xi_k^{m})$ therefore represents the measurement noise. Hence, Lemmas~\ref{lem:j0}-\ref{lem:PropRV} remain valid when the measurement noise is also considered.
\end{exmp}
\section{Numerical example} \label{sec:Simulation}
We consider a continuous-time LTI helicopter model from \citet{keel88robust} and discretize it with sampling time $t_s=\SI{0.5}{\second}$.\footnote{The code for the numerical example is available at: \url{https://github.com/OptCon/StochWillemsLemma}. The proposed algorithm will be integrated into the toolbox \texttt{PolyOCP.jl} \citep{ou26polyocp}.} The system is unstable but controllable and observable. The VARX matrices in~\eqref{eq:VARX} are
\begin{align*}
		\hat{A} &= \begin{bmatrix*}[r]
		6.1223 &\quad -0.6048 &\quad -5.2759	&\quad \phantom{-}1.1125\\
		46.1936 &\quad -4.4954 &\quad -47.1983 &\quad	\phantom{-}8.7976
	\end{bmatrix*},\\
	\hat{B} &= \begin{bmatrix*}[r]
		-0.0250	&\quad 3.6432 &\quad \phantom{-0}0.2664 &\quad 0.0368\\
		0.0089 &\quad 27.1174 &\quad \phantom{-0}1.7629 &\quad -3.2636
	\end{bmatrix*}.
\end{align*}

The elements of $V_k$, $k\in \N$ are independently and uniformly distributed with $V^i\sim\mcl{U}(-0.1,0.1)$ for both $i=1,2$, where $V^i$ denotes the $i$-th element of $V$. Then we solve the following OCP
\begin{subequations}
	\begin{align}
		\min_{U_k, Y_k,G^{w_k}, g, k\in\I_{[1,N]}}
		\sum_{k=1}^N\mean[Y_k^\top QY_k& + U_k^\top R U_k]\\
		\text{s.t.}\quad \text{dynamics in form }&\text{of Lemma~\ref{lem:PropRV}},\\
		\mbb{P}[-5\leq U_k^1\leq 5]\geq &0.8, k\in\I_{[1,N]}, \label{eq:cc_u1}\\
		\mbb{P}[-4\leq U_k^2\leq 4]\geq &0.8, k\in\I_{[1,N]} \label{eq:cc_u2}
	\end{align}
\end{subequations}
in the PCE framework. The weighting matrices in the stage cost are $Q=R=I_2$. The chance constraints~\eqref{eq:cc_u1}-\eqref{eq:cc_u2} imposed on $U^1$ and $U^2$ individually are implemented via a conservative reformulation \citep{calafiore06distributionally}. The computations are done on a virtual machine with an AMD Ryzen 9 3900X 12-Core Processor with 3.79 GHz, 64 GB of RAM in \texttt{julia} using \texttt{IPOPT} \citep{waechter06implementation}.

\begin{figure}[t]
	\begin{subfigure}{\linewidth}
		\centering
		\includegraphics[width=0.85\linewidth,trim={36mm 14mm 22mm 30mm},clip]{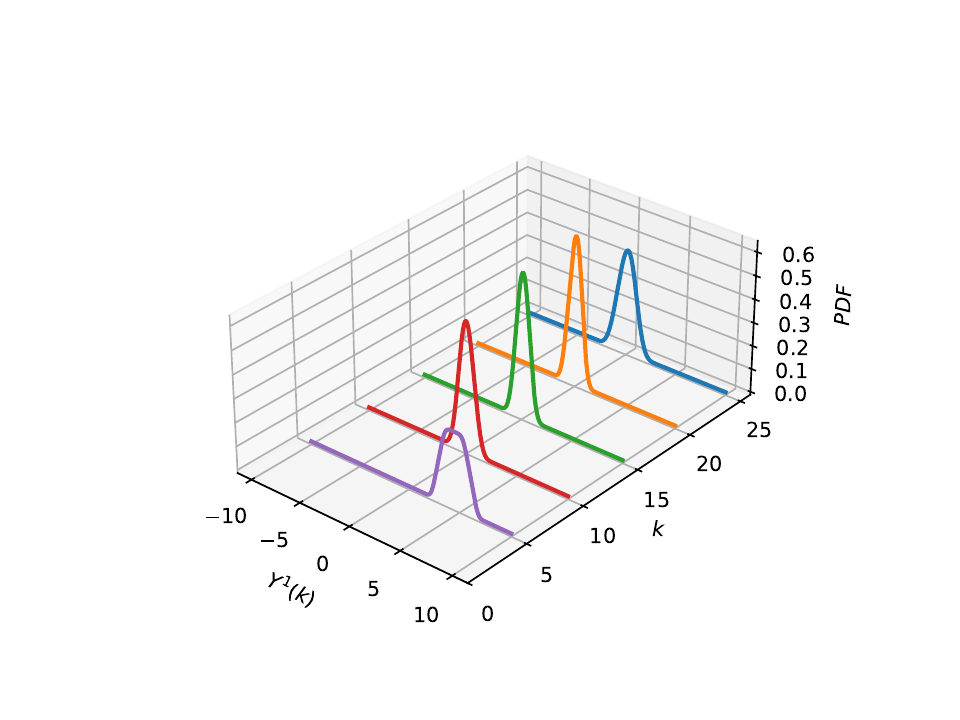}
	\end{subfigure}
	\begin{subfigure}{\linewidth}
		\centering
		\includegraphics[width=0.85\linewidth,trim={36mm 14mm 22mm 30mm},clip]{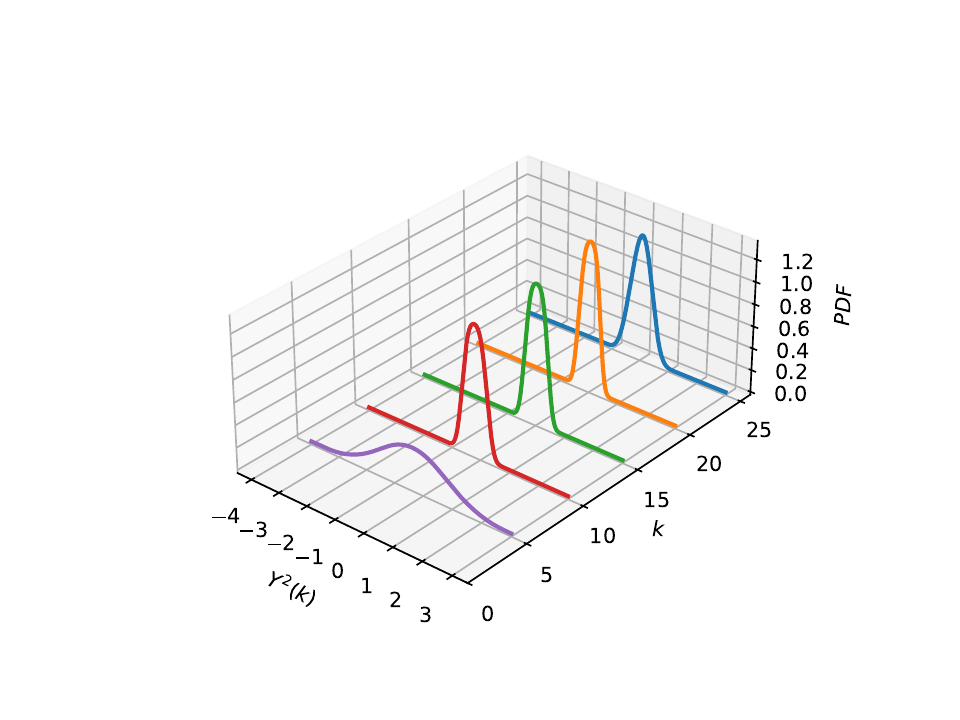}
	\end{subfigure}
	\caption{Evolution of the PDFs of the outputs $Y^1$ and $Y^2$ over horizon $N=25$.}
	\label{fig:DistributionY}
\end{figure}

\begin{figure}[t]
	\begin{center}
		\includegraphics[width=1\linewidth]{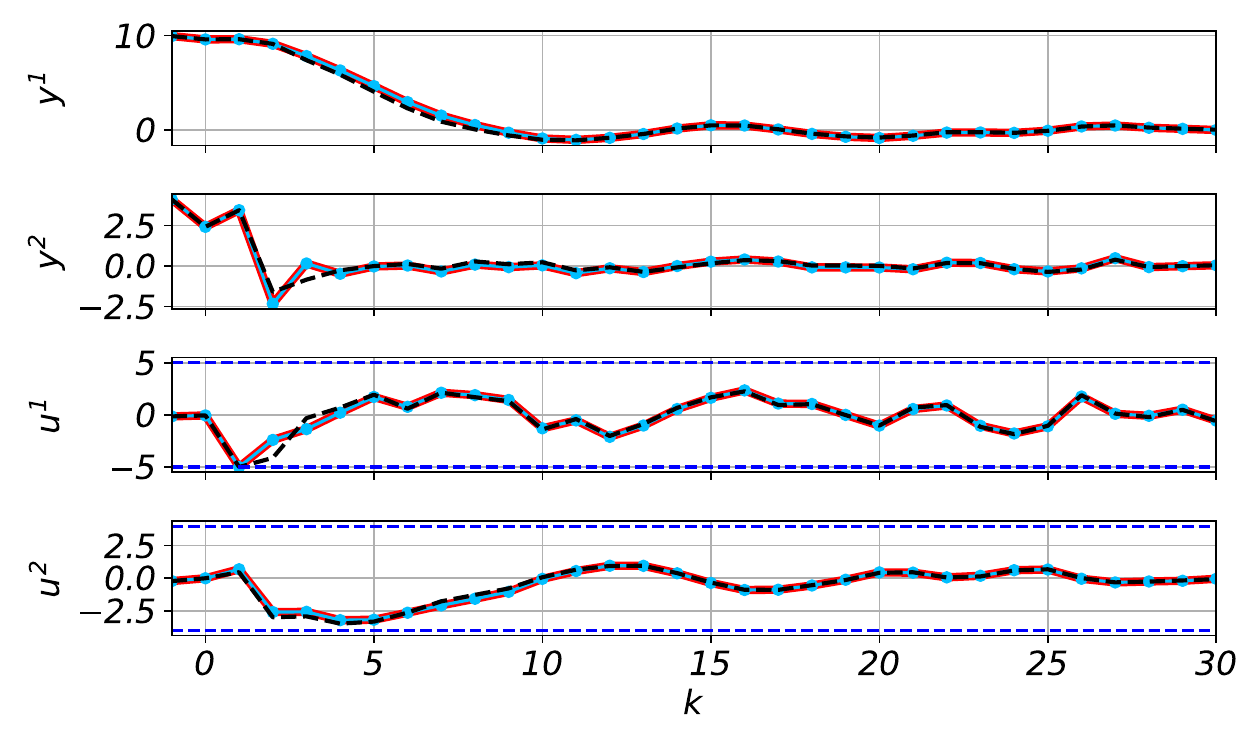}
		\caption{Helicopter example with uniformly distributed disturbance. Red-solid line: Scheme~I; blue-solid line with circle marker: Scheme~II; black-dashed line: Scheme~III; deep blue-dashed line: Chance constraints.} \label{fig:TrajComparison}		
	\end{center}
\end{figure}

First we solve the above stochastic OCP with prediction horizon $N=25$ and a randomly sampled initial condition around output $[10,10]^\top$. The evolution of the Probability Density Functions (PDF) of the system outputs $Y^1$ and $Y^2$ for the computed optimal input is depicted in Figure~\ref{fig:DistributionY}. As one can see, $Y^1$ and $Y^2$ both show narrow distributions close to 0 in the middle of the horizon.

Then we compare the following three schemes in closed loop with prediction horizon $N=10$ at each time step:
\newpage
\begin{itemize}
	\item[I)] Lemma~\ref{lem:PropRV} with data $(u,y)^{\ud}$ of system~\eqref{eq:VARXFree},
	\item[II)] Lemma~\ref{lem:StochFundam} with data $(u,y,v)^{\da}$ of system~\eqref{eq:VARXReal},
	\item[III)] Lemma~\ref{lem:PropRV} with data $(u,y)^{\da}$ of system~\eqref{eq:VARXReal}.
\end{itemize}
Note that we omit the comparison of the above schemes with robust schemes since this paper mainly focuses on predicting the future trajectories of stochastic LTI systems. Nevertheless, \citet{pan25data} have compared Scheme~II with a distributionally robust scheme. Note that to apply Scheme III using data of~\eqref{eq:VARXReal}, we first estimate the disturbance realizations and then generate a corresponding undisturbed trajectory $(u,y)^{\ud}$ of~\eqref{eq:VARXFree} via~\eqref{eq:PredictorModify}. We use the procedure proposed in Section~\ref{sec:Estimation} to find stabilizing feedback controllers for all of the above schemes and then record/generate data of length 80 to construct all Hankel matrices. The 2-norm condition numbers of the Hankel matrices for Schemes~I-III are 98.1, 686.0, and 43.7, respectively. The condition number of the Hankel matrix for Scheme~II is higher than the others due to the disturbances. Given an initial condition around $[10,10]^\top$ and the same disturbance realizations, we compute the closed-loop trajectories for 30 time steps for Schemes~I-III, see Figure~\ref{fig:TrajComparison}. For a detailed closed-loop algorithm in the data-driven setting we refer to Algorithm~1 of~\citet{pan23stochastic}.
Moreover, we compute the closed-loop trajectory for the model-based approach with an exact model. The resulting trajectory differs from that of Scheme~I by at most $3.974\cdot 10^{-4}$. Therefore, we use the solution of Scheme~I as the baseline for comparison and do not include the trajectory of the model-based approach in Fig.~\ref{fig:TrajComparison}.

\begin{table}[t!]
	\caption{Comparison of the data amount in Hankel matrices and the computation time for 1000 samplings.}
	\label{tab:ComparisonTime}
	\centering
	\begin{adjustbox}{width=1\columnwidth,center}
		\begin{tabular}{cccccc}
			\toprule
			\multirow{2}{*}{\shortstack{\\ \\ Data-driven\\ scheme}} &  \multirow{2}{*}{\shortstack{\\ \\ Number of \\ entries in $\Hankel$}}&\multicolumn{2}{c}{Computation time}  & 	\multirow{2}{*}{$J^\text{cl}$ $[-]$} & \multirow{2}{*}{\shortstack{\\ \\ Computational \\ Complexity [\unit{\micro\second}]}} \\
			\cmidrule(lr){3-4} & & Mean [\unit{\second}] & SD [\unit{\second}]   &\\
			\midrule
			I & 19044 &0.434 & 0.082  & 15.500 & 31.384\\
			II & 69552 & 0.830 & 0.069  & 15.500 & 113.334\\
			III & 19044 & 0.540 & 0.101 & 16.191 & 31.634 \\
			\bottomrule
		\end{tabular}
	\end{adjustbox}
\end{table}

We observe that the input-output trajectories for Scheme~I and Scheme~II are identical with a maximal difference of $1.413\cdot 10^{-5}$, while Scheme~III results in a slightly different trajectory due to the estimation error of disturbance realizations. We sample 1000 sequences of initial condition and disturbance realizations and summarize the computation time in Table~\ref{tab:ComparisonTime}, where SD refers to standard deviation. One can see that Scheme~I reduces computation time by an average of 47.6\% in comparison to Scheme~II with a similar standard deviation, while the average closed-loop stage cost $J^{\text{cl}}$ remains the same.
Additionally, the computation time required to find the corresponding $G$ for given initial conditions is used as an indicator of the computational complexity of Lemmas~\ref{lem:StochFundam} and~\ref{lem:PropRV}. Note that the resulting $G$ is not necessarily a feasible solution of the OCP, since the chance constraints are not taken into account. We record the computation time once at every time step in closed loop and then report the average computation time in Table~\ref{tab:ComparisonTime}. The computation time of Lemma~\ref{lem:PropRV} is reduced by 72.3\% compared to that of Lemma~\ref{lem:StochFundam}, which is consistent with the 86.3\% reduction in computational complexity suggested by~\eqref{eq:CC}, and exceeds the minimum reduction of 66.7\% for $\dimy=\dimu$ as stated in Remark~\ref{rem:Computation}. The reduction in computation time is slightly smaller than that in computational complexity by~\eqref{eq:CC}, since the computation time is also influenced by other factors, such as the numerical properties of the Hankel matrices and memory allocation overhead.
Moreover, the performance of Scheme~I and Scheme~III reports similar results in terms of the computation time, while Scheme~III has a slightly higher average cost due to the estimation error of disturbances. Importantly, by using Lemma~\ref{lem:PropRV} in Scheme~I, one can save 72.6\% amount of the data in the Hankel matrix in this simulation example compared to Scheme~II without loss of performance.
\section{Conclusion} \label{sec:Conclusion}
This paper has proposed a stochastic variant of the fundamental lemma towards stochastic LTI systems with process disturbances in the PCE framework. Based on the superposition principle, the stochastic LTI system is decoupled into a number of subsystems, each of which corresponds to a source of uncertainty, e.g., a disturbance at certain time step. Then given the causality constraints on the inputs and outputs, the additive disturbances can be converted into output initial conditions of the decoupled subsystems. This way, the dynamics of the subsystems do not contain any disturbances except for the initial conditions. Therefore, in contrast to our previous work, i.e. Lemma~\ref{lem:StochFundam}, the proposed variant of the fundamental lemma does not require any disturbance data in Hankel matrices. As a by-product we have shortened the prediction horizon for the PCE coefficients related to the disturbances, which leads to an acceleration in numerical implementations. At the same time, the application of this proposed variant of the fundamental lemma for control and a detailed comparison to other schemes including robust control remain open. Future work will also consider disturbance estimation with error bounds, the extension to nonlinear systems, and the fast computation of data-driven stochastic optimal control.

\section*{Acknowledgements}
The authors would like to thank Jonas Schie{\ss}l, Michael Heinrich Baumann, and Lars Gr{\"u}ne at the Mathematical Institute, University of Bayreuth, Germany, for their valuable feedback. The authors acknowledge funding by the Deutsche Forschungsgemeinschaft (DFG, German Research Foundation) - project number 499435839; 520388526.
                         
\bibliography{arXivRef}
\end{document}